\documentclass[twocolumn]{autart} 

\usepackage{amssymb}
\usepackage{cite}
\usepackage{amsmath}
\usepackage{cases}
\usepackage{epsfig}
\usepackage{graphicx}
\usepackage{epstopdf}
\usepackage{subfig}
\usepackage{mathrsfs}

\usepackage{float}
\usepackage{multirow}
\usepackage{graphicx}
\usepackage{colortbl}

\floatstyle{ruled}
\newfloat{algorithm}{tbp}{loa}
\providecommand{\algorithmname}{Algorithm}
\floatname{algorithm}{\protect\algorithmname}
\usepackage{algorithm}
\usepackage{algorithmic}
\usepackage{color}

\usepackage{setspace}
\newcommand{\be}{\begin{equation}}
\newcommand{\ee}{\end{equation}}
\newcommand{\bea}{\begin{eqnarray}}
\newcommand{\eea}{\end{eqnarray}}
\newcommand{\nn}{\nonumber}

\newtheorem{theorem}{Theorem}
\newtheorem{assumption}{Assumption}
\newtheorem{lemma}{Lemma}
\newtheorem{definition}{Definition}
\newtheorem{proof}{Proof}

\def\ba{\begin{array}}
\def\ea{\end{array}}
\def\a{\alpha}

\def\m{\mu}

\def\ge{\geqslant}

\begin{document}

\begin{frontmatter}
\title{Distributed Weighted Least-squares Estimation for Networked Systems with Edge Measurements\thanksref{footnoteinfo}}
\thanks[footnoteinfo]{This paper was not presented at any IFAC
meeting. This work was supported by the National Natural Science Foundation of China (Grant Nos.~61633014, 61803101 and U1701264).}
\author[SD]{Qiqi Yang}\ead{qiqiyang2016@163.com},
\author[newcastle]{Zhaorong Zhang}\ead{zhaorong.zhang@uon.edu.au},
\author[newcastle,GG]{Minyue Fu}\ead{minyue.fu@newcastle.edu.au}
\corauth[cor]{Corresponding author: Zhaorong Zhang, Tel. +61-408528333.}
\address[SD]{School of Control Science and Engineering, Shandong University, Jinan 250061, China.}       
\address[newcastle]{School of Electrical Engineering and Computing, The University of Newcastle, NSW 2308, Australia.}        
\address[GG]{School of Automation, Guangdong University of Technology, 
Guangzhou 510006, China.}       

\begin{abstract}
This paper studies the problem of distributed weighted least-squares (WLS) estimation for an interconnected linear measurement network with additive noise. Two types of measurements are considered: self measurements for individual nodes, and edge measurements for the connecting nodes. Each node in the network carries out distributed estimation by using its own measurement and information transmitted from its neighbours. We study two distributed estimation algorithms: a recently proposed distributed WLS algorithm and the so-called Gaussian Belief Propagation (BP) algorithm. We first establish the equivalence of the two algorithms. We then prove a key result which shows that the information matrix is always generalised diagonally dominant, under some very mild condition. Using these two results and some known convergence properties of the Gaussian BP algorithm, we show that the aforementioned distributed WLS algorithm gives the globally optimal WLS estimate asymptotically. A bound on its convergence rate is also presented.
\end{abstract}

\end{frontmatter}

\section{Introduction}
As the applications for large-scale networked systems increase rapidly, 
distributed estimation algorithms for such systems is essential, and they are widely 
applied to sensor networks \cite{kar2012,lijl2007}, networked linear systems \cite{mou2015}, network-based state estimation \cite{damian2015}, multi-agent systems \cite{lin2014,lin2016}, multi-agent optimization \cite{nedic2009},  and so on.

In this paper, we are  interested in a distributed algorithm recently proposed in \cite{damian2015} (Algorithm 4 in \cite{damian2015}) to solve weighted least-squares (WLS) estimation for  large-scale networked systems. This algorithm is fully distributed and iterative. It was proved in \cite{damian2015} that this distributed algorithm produces the exact WLS solution (i.e., the globally optimal estimate) after a finite number of iterations, if the network graph is acyclic. For a general network graph, many simulations suggest that the distributed WLS algorithm in \cite{damian2015} is capable to generate the exact WLS solution asymptotically, although the theoretical verification is lacking. {\em The purpose of this paper is to analyze the convergence property of this distributed WLS algorithm for a class of general network graphs.}

Another pertinent distributed algorithm comes from seemingly unrelated field of stochastic learning,  used to compute the conditional means and variances from a large-scale Gaussian random field. This algorithm is known as Gaussian Belief Propagation algorithm \cite{weiss2001}, a variant of the celebrated Belief Propagation (BP) algorithm originally proposed by Pearl \cite{pearl1988} in 1988.  

We consider the {\em distributed WLS estimation problem} for an interconnected linear measurement network with additive noise. Each node in the network has an unknown variable. The available measurements can be divided into two types: 1) {\em self measurement} for an individual node, which involves the node variable only, and 2) {\em edge measurement} for an edge,  which involves the two joining nodes. The contributions of this paper are as follows:
\vspace{-4mm}
\begin{itemize}
\item Firstly, we compare the distributed WLS algorithm (Algorithm 4 in \cite{damian2015}) with the Gaussian Belief Propagation (BP) algorithm which is expressed using the information matrix of the measurement system and we establish their equivalence. 
\item We then prove a key result for the case of scalar variables to show that the information matrix is always generalised diagonally dominant, under some very mild condition. 
\item Using these two results and some known convergence properties of the Gaussian BP algorithm, we present several convergence results for the distributed WLS algorithm. For an acyclic graph with vector variables, the algorithm gives the globally optimal WLS estimate in a finite number of iterations.  For a cyclic graph with scalar variables, the algorithm gives the globally optimal estimate asymptotically. Moreover, a bound on its convergence rate is also provided.  
\end{itemize}

\vspace{-2mm}\section{PROBLEM FORMULATION}
Consider a measurement network with $n$ nodes with an associated graph $\mathcal{G}=(\mathcal{V},\mathcal{E})$ with $\mathcal{V}=\{1,\cdots,n\}$ and $\mathcal{E}\subset \mathcal{V}\otimes\mathcal{V}$. We use $\mathcal{N}_i=\{j:(i,j)\in \mathcal{E}\}$ to denote the set of neighbours of node $i$.  The {\em distance} between two nodes $i,j\in \mathcal{V}$ is the length of the shortest path between the two nodes. The {\em diameter} of $\mathcal{G}$ is the largest distance between any two nodes in $\mathcal{V}$. The graph $\mathcal{G}$ is known as the {\em measurement graph} and {\em communication graph}.

For each node $i\in\mathcal{V}$, $x_i\in \mathbb{R}^{n_i}$ denotes the {\em state} (or {\em variable}) of node $i$, and its measurements can be divided into two types: a {\em self measurement} $z_{i}$ involving $x_i$ only, and an {\em edge measurement} $z_e$ for each  $e=(i,j)$ involving both $x_i$ and $x_j$. These measurements are described by
\begin{align}
z_{i}&=A_{i}x_i+v_{i},\label{eq:self}\\
z_e&=B_{ij}x_i+B_{ji}x_j+v_e,\label{eq:edge}
\end{align}
where $v_{i}$ is the self measurement noise with normal distribution $\mathbf{N}(0,R_i)$, $v_e$ is the edge measurement noise (for edge $e$) with normal distribution $\mathbf{N}(0,R_e)$, $R_i>0$ and $R_e>0$ are covariances, and $A_i, B_{ij}$ and $B_{ji}$ are matrices of appropriate dimensions. The noises $v_{i}$ and $v_e$ are statistically independent, whenever $i\neq j$. Similarly, $v_{e_1}$ and $v_{e_2}$ are statistically independent for $e_1\neq e_2$.  Each self measurement (\ref{eq:self}) is known to node $i$ only, and each edge measurement (\ref{eq:edge}) is known to both nodes $i$ and $j$. 

Let the order of the nodes in $\mathcal{V}$ be $1,2,\cdots, n$ and the order of the edges in $\mathcal{E}$ be $e_1, e_2, \ldots, e_p$. Define 
\begin{align*}
x&=\mathrm{col}\{x_1,\cdots, x_n\},\nn\\
z &= \mathrm{col}\{z_1,\cdots, z_n, z_{e_1}, \cdots, z_{e_p}\}, \nn\\
v&= \mathrm{col}\{v_1,\cdots, v_n, v_{e_1}, \cdots, v_{e_p}\}, \nn\\
R&=\mathrm{diag}\{R_1, \cdots, R_n,R_{e_1}, \cdots,R_{e_p}\},\nn\\
H_i &= [0\ \cdots  \ 0\  A_i \ 0 \cdots 0], \ \forall i\in \mathcal{V}, \nn 
\end{align*}
\begin{align}
H_e &= [0\ \cdots  \ 0\  B_{ij} \ 0 \cdots 0 \ B_{ji} \ 0 \cdots 0], \ \forall e\in \mathcal{E},\nn \\
H&=\mathrm{col}\{H_1, \cdot, H_n, H_{e_1}, \cdots, H_{e_p}\}. 
\end{align}
We can rewrite the whole measurement model as
\begin{align}
z&=Hx+v. \label{eq0}
\end{align}

\vspace{-2mm}\begin{rem}
{\rm The above measurement models are widely used in practice. Self measurements are typically used to measure local variables such as temperature at a local point, absolute position of a sensor, voltage or current at a nodal point in a power network. Edge measurements can be used to measure {\em relative information} such as relative position, angle or velocity between two drones, pressure drop between two taps, and more subtle examples like current through a power branch.}
\end{rem}

The WLS estimate $x^{\star}$ of $x$ is defined to be
\begin{align}
x^{\star}&=\arg\min_{x}(z-Hx)^TR^{-1}(z-Hx). \label{eq:wls}
\end{align}
This can be rewritten as $\min_{x}(x^T\Psi x + 2 \alpha^T x)$
with
\begin{align}\label{eq2}
\Psi&=H^{T}R^{-1}H,\ \ \ 
\alpha=H^{T}R^{-1}z,
\end{align}
 and the solution can be given by
\begin{align}
\label{eq1} x^{\star}&=\Psi ^{-1}\alpha, 
\end{align}
which will be called the {\em globally optimal solution}. We stress that the WLS probem uses the measurement error covariances as the weighting matrices, which makes the solution optimal in the maximum likelihood sense. But this optimality relies on the accuracy of the covariances.

The goal for a {\em distributed WLS} solution is to derive a distributed algorithm in which node $i$ computes {\em only} the $i$-th component $x_i^{\star}$ of $x^{\star}$ in an iterative fashion using only the locally available measurements $z_{i,i}$ and $z_{(i,j)},  j\in\mathcal{N}_i$, and information exchange with its neighbouring nodes. 

\begin{assumption}
 \label{a1} The measurement graph $\mathcal{G}$ is connected.
 \end{assumption}
\begin{assumption} \label{a2} The measurement system $(\ref{eq0})$ has at least one self measurement for some nodes in the graph.
 \end{assumption}

\begin{rem}
{\rm Assumption~\ref{a2} permits the estimation problem to have a unique solution.}
\end{rem}

\begin{definition} \label{def1} \cite{berman1994} Denote $Z^{n\times n}=\{A=(a_{ij})\in \mathbb{R}^{n\times n}: a_{ij}\leq 0,i\neq j\}$.
A matrix $A\in Z^{n\times n}$ is an M-matrix if it can be expressed by $A=sI-B$, where $B=(b_{ij})$ with $b_{ij} \geq0$ for all $1 \leq i,j \leq n$ and $s\geq\rho(B)$. 
The comparison matrix of $A\in \mathbb{R}^{n\times n}$, denoted by $\bar{A}=(\alpha_{ij})\in Z^{n\times n}$, is given by
$\a_{ii}=|a_{ii}|$ and $\a_{ij}=-|a_{ij}|$ for all $i$ and $j\ne i$. 
\end{definition}

\begin{definition} \label{def2} A matrix $A\in \mathbb{R}^{n\times n}$ is said to be {\em generalized diagonally dominant} if there exists a diagonal matrix $D=\mathrm{diag}\{d_i\}$ with all  $d_i>0$ such that $AD$ is diagonally dominant, i.e.,
$|a_{ii}|d_i>\sum_{j\neq i}|a_{ij}|d_j$ for all $i$.
\end{definition}

\section{Main Results}

In this section, we present two distributed algorithms,  establish their equivalence and study the convergence properties of the distributed WLS algorithm.

\begin{algorithm}[ht]
\caption{Distributed WLS Algorithm \cite{damian2015} (using self and edge measurements directly)}
\begin{itemize}
\item \textbf{Initialization:} For each $i=1,\cdots,n$, node $i$ computes Compute $\Psi_{ii}$, $\alpha_i$ and $\hat{x}_i(0)=\Psi_{ii}^{-1}\alpha_i$, and transmits to each $j\in \mathcal{N}_i$ the initial messages:
\begin{align}
\Sigma_{i\rightarrow j}(0)&=\Psi_{ii}^{-1}, \ 
x_{i\rightarrow j}(0)=\Sigma_{i\rightarrow j}(0)\alpha_i. \label{eq:x0}
\end{align}
\item  \textbf{Main loop:} For $t=1,2,\cdots$, each node $i$ computes
\begin{align}
\hat{\Psi}_i(t)&=\Psi_{ii}-\sum_{v\in \mathcal{N}_i}\Psi_{vi}^{T}\Sigma_{v\rightarrow i}(t-1)\Psi_{vi}, \label{eq:temp1}\\
\hat{\alpha}_i(t)&=\alpha_i-\sum_{v\in \mathcal{N}_i}\Psi^{T}_{vi}x_{v\rightarrow i}(t-1), \label{eq:temp2}\\
\hat{x}_i(t)&=\hat{\Psi}^{-1}_i(t)\hat{\alpha}_i(t), \label{xxx}
\end{align}
 then, for each $j\in \mathcal{N}_i$, computes the new messages 
\begin{align}
\Sigma_{i\rightarrow j}(t)&=\left (\hat{\Psi}_i(t)+\Psi^{T}_{ji}\Sigma_{j\rightarrow i}(t-1)\Psi_{ji}\right )^{-1}, \label{eq:temp3}\\
x_{i\rightarrow j}(t)&=\Sigma_{i\rightarrow j}(t)\left (\hat{\alpha}_i(t)+\Psi^{T}_{ji}x_{j\rightarrow i}(t-1)\right ), \label{eq:temp4}
\end{align}
and transmits them to node $j$.
\end{itemize}
\label{al1}
\end{algorithm}

{\bf Distributed WLS Algorithm}

\vspace{-2mm}Algorithm $\ref{al1}$ is simplified from Algorithm 4 in \cite{damian2015} for the case with self and edge measurements only.
Denote $\alpha=\mathrm{col}\{\alpha_i\}$ and $\Psi=(\Psi_{ij})$. For $i\in \mathcal{V}$ and $j\in \mathcal{N}_i$, 
\begin{align}
\alpha_i &=A_i^{T}R_{i}^{-1}z_{i}+\sum_{j\in \mathcal{N}_i}B_{ij}^{T}R_{(i,j)}^{-1}z_{(i,j)}, \label{eq:alphai}\\
\Psi_{ii}&\hspace{-1mm}=A_i^{T}R_{i}^{-1}A_i+\hspace{-2mm}\sum_{j\in\mathcal{N}_i}B_{ij}^{T}R_{(i,j)}^{-1}B_{ij}, \ \Psi_{ij}\hspace{-1mm}=B_{ij}^{T}R_{(i,j)}^{-1}B_{ji}. \nonumber
\end{align}
The basic idea of the algorithm is as follows. In iteration $t=0$, each $i\in \mathcal{V}$ computes its $\Psi_{ii}$ and $\alpha_{i}$ locally. For each $j\in \mathcal{N}_i$, messages passed from $i$ to $j$ $\Psi_{i\rightarrow j}(0)$ and $\alpha_{i\rightarrow j}(0)$. At iteration $t>0$, $\hat{\Psi}_{i}(t)$, $\hat{\alpha}_i(t)$ and $\hat{x}_i(t)$ (the estimate of $x_i^{\star}$) are calculated  using $\Psi_{ii}$ and $\alpha_i$ and messages sent by $j\in \mathcal{N}_i$ in the previous iteration. Node $i$ uses all the messages sent from $v\in \mathcal{N}_i$ except $j$ to calculate $\Psi_{i\rightarrow j}(t)$ and $\alpha_{i\rightarrow j}(t)$, and sends them to node $j$. 

\begin{algorithm}[ht]
\caption{Gaussian BP Algorithm \cite{weiss2001} (based on $\Psi$)}
\begin{itemize}
\item \textbf{Initialization:}  For each $i=1,\cdots,n$, node $i$ computes $P_{ii}=A_i^TR_i^{-1}A_i$ and $\mu_{ii}=P_{ii}^{-1}\alpha_i$ using (\ref{eq:alphai}),  and transmits to each $j\in \mathcal{N}_i$ the initial messages $P_{i\rightarrow j}(0)=\Gamma_{ji}$ and 
$\m_{i\rightarrow j}(0)=0.$
\item  \textbf{Main loop:} For $t=1,2,\cdots$, each node $i$ computes
\begin{align}
\hspace{-4mm}P_i(t)&=P_{ii}+\sum_{v\in \mathcal{N}_i}P_{v\rightarrow i}(t-1), \label{eq5}\\
\hspace{-4mm}\m_i(t)&=P_i^{-1}(t)[P_{ii}\m_{ii}+\sum_{v\in \mathcal{N}_i}P_{v\rightarrow i}(t-1)\m_{v\rightarrow i}(t-1)], \label{eq6}
\end{align}
also, for each $j\in \mathcal{N}_i$, computes the new messages:\vspace{-1mm}
\begin{align}
\hspace{-4mm}P_{i\rightarrow j}(t)&=\Gamma_{ji}-\Psi_{ji}(\Gamma_{ij}+P_0(t))^{-1}\Psi_{ij},\label{eq7}\\
\hspace{-4mm}\m_{i\rightarrow j}(t)&=-P_{i\rightarrow j}^{-1}(t)\Psi_{ji}(\Gamma_{ij}+P_0(t))^{-1}P_0(t)\mu_0(t) \label{eq8} 
\end{align}
with\vspace{-1mm}
\begin{align} 
P_0(t)&=P_{ii}+\sum_{v\in \mathcal{N}_i\setminus j}P_{v\rightarrow i}(t-1), \label{eq9}\\
\mu_0(t)&=P_0^{-1}(t)[P_{ii}\m_{ii}+\hspace{-2mm}\sum_{v\in \mathcal{N}_i\setminus j}P_{v\rightarrow i}(t-1)\m_{v\rightarrow i}(t-1)], \label{eq10}
\end{align}
and transmits them to node $j$.
\end{itemize}
\label{al2}
\end{algorithm}

{\bf Gaussian BP Algorithm}

The Gaussian BP algorithm in \cite{weiss2001}  is shown in Algorithm~\ref{al2}. In the algorithm, the estimate of $x_i^{\star}$ at iteration $t$ will be given by $\mu_i(t)$. The algorithm is developed based on the joint distribution model for $x$. For the measurements (\ref{eq:self})-(\ref{eq:edge}), this is proportional to 
 \begin{align*}
&\prod_{i\in \mathcal{V}} \exp(-\frac{1}{2}\varepsilon_i^TR_i^{-1}\varepsilon_i)\prod_{(i,j)\in \mathcal{E}} \exp(-\frac{1}{2}\varepsilon_{(i,j)}^TR_{(i,j)}^{-1}\varepsilon_{(i,j)}) , 
\end{align*}
 where $\varepsilon_{(i,j)} = z_{(i,j)}-B_{ij}x_i-B_{ji}x_j$,
$ \varepsilon_i = z_i-A_ix_i.$ Rewriting the above gives 
 \begin{align*}
  & \prod_{i\in \mathcal{V}}\exp(-\frac{1}{2} x_i^TP_{ii} x_i-\alpha_i^Tx_i)\hspace{-3mm}\prod_{(i,j)\in \mathcal{E}}\hspace{-2mm} \exp(-\frac{1}{2}[x_i^T  x_j^T]V_{ij}\hspace{-1mm}\left [\begin{array}{c}x_i \\ x_j \end{array} \right ]\hspace{-1mm}), 
\end{align*}
where $P_{ii} = A_i^TR_i^{-1}A_i$, $\alpha_i$ is given in (\ref{eq:alphai}) and 
\begin{align*}
V_{ij} &= \left [\begin{array}{c} B_{ij}^T \\ B_{ji}^T\end{array}\right ]R_{(i,j)}^{-1}[B_{ij} \ \ B_{ji}] = \left [\begin{array}{cc} 
\Gamma_{ij} & \Psi_{ij} \\ \Psi_{ji} & \Gamma_{ji} \end{array}\right ]
\end{align*}
with $\Gamma_{ij}=B_{ij}^TR_{(i,j)}^{-1}B_{ij}$ and $\Gamma_{ji}=B_{ji}^TR_{(i,j)}^{-1}B_{ji}$.

{\bf Convergence Properties}

\vspace{-2mm}Our first result below compares these two algorithms.

\vspace{-2mm}\begin{theorem}
 \label{prop1} Algorithm \ref{al1} and Algorithm \ref{al2} are equivalent  in the sense that
 \begin{align}
 \hat{\Psi}_i(t) &= P_i(t+1); \ \ \hat{x}_i(t) = \mu_i(t+1), \ \ \forall t\ge0. \label{eq:eqiv}
 \end{align}
 \end{theorem}
 
\vspace{-2mm}\begin{proof}
We first claim that the following two equations hold  for any $t=0,1,2,\ldots$: 
\begin{align}
\Sigma_{i\rightarrow j}(t) &= (\Gamma_{ij}+P_0(t+1))^{-1}, \label{eq:temp11} \\
\Sigma_{i\rightarrow j}^{-1}(t) x_{i\rightarrow j}(t) &= P_0(t+1)\mu_0(t+1). \label{eq:temp12}
 \end{align}
 Proceed by induction. For $t=0$, from (\ref{eq9}), we see that 
 \begin{align*}
  P_0(1) &= P_{ii} + \sum_{v\in \mathcal{N}_i\backslash j} P_{v\rightarrow i}(0) = P_{ii}+ \sum_{v\in \mathcal{N}_i\backslash j} \Gamma_{iv}. 
  \end{align*}
 It follows from the above that 
  \begin{align*}
\Gamma_{ij}+P_0(1)  &= P_{ii} + \sum_{v\in \mathcal{N}_i} \Gamma_{iv} \\
&=A_i^TR_i^{-1}A_i + \sum_{v\in \mathcal{N}_i} B_{iv}^TR_{(v,i)}^{-1}B_{iv} =\Psi_{ii}.
\end{align*}
Using (\ref{eq:x0}), the above implies that (\ref{eq:temp11}) holds for $t=0$.  Similarly, using (\ref{eq10}) and (\ref{eq:x0}), we get
\begin{align*}
P_0(1)\mu_0(1) &= P_{ii}\mu_{ii} + \sum_{v\in \mathcal{N}_i\backslash j} P_{v\rightarrow i}(0)\mu_{v\rightarrow i}(0)\\
&= P_{ii}\mu_{ii} = \alpha_i=\Sigma_{i\rightarrow j}^{-1}(0)x_{i\rightarrow j}(0),
\end{align*}
which confirms (\ref{eq:temp12}) for $t=0$. 
Now, suppose (\ref{eq:temp11})-(\ref{eq:temp12}) holds for some $t=k, k\ge0$.  
From (\ref{eq7}) and (\ref{eq5}), we get
\begin{align*}
P_0(k+2)&= P_{ii}+\hspace{-2mm}\sum_{v\in \mathcal{N}_i \backslash j} P_{v\rightarrow i}(k+1) \\
&= P_{ii}+\hspace{-2mm}\sum_{v\in \mathcal{N}_i \backslash j}\Gamma_{iv} - \sum_{v\in \mathcal{N}_i \backslash j} \Psi_{iv}\Sigma_{v\rightarrow i}(k)\Psi_{vi}.
\end{align*}
It follows that 
\begin{align*}
\Gamma_{ij}+P_0(k+2) &= P_{ii} + \sum_{v\in \mathcal{N}_i} \Gamma_{iv} -\hspace{-2mm}\sum_{v\in \mathcal{N}_i \backslash j} \Psi_{vi}^T\Sigma_{v\rightarrow i}(k)\Psi_{vi} \\
&= \Psi_{ii}-\sum_{v\in \mathcal{N}_i \backslash j} \Psi_{vi}^T\Sigma_{v\rightarrow i}(k)\Psi_{vi} \\
&=\Sigma_{i\rightarrow j}^{-1}(k+1),
\end{align*}
which verifies (\ref{eq:temp11}) for $t=k+1$.  
Next, using (\ref{eq8}), we get
\begin{align*}
&P_{i\rightarrow j}(k+1)\mu_{i\rightarrow j}(k+1)\\
 =& -\Psi_{ji}(\Gamma_{ij}+P_0(k+1))^{-1}P_0(k+1)\mu_0(k+1)  \\
 =& -\Psi_{ji}\Sigma_{i\rightarrow j}(k)\Sigma_{i\rightarrow j}^{-1}(k)x_{i\rightarrow j}(k) 
 = - \Psi_{ji}x_{i\rightarrow j}(k).
 \end{align*}
Using (\ref{eq10}) and the above, we have
\begin{align*}
&P_0(k+2)\mu_0(k+2) \\
=&P_{ii}\mu_{ii} + \sum_{v\in \mathcal{N}_i\backslash j} P_{v\rightarrow i}(k+1)\mu_{v\rightarrow i}(k+1) \\
=& \alpha_i -  \sum_{v\in \mathcal{N}_i\backslash j} \Psi_{vi}^Tx_{v\rightarrow i}(k)\\
=& \hat{a}_i(k+1) + \Psi_{ji}^T x_{j\rightarrow i}(k) 
=\Sigma_{i\rightarrow j}^{-1}(k+1)x_{i\rightarrow j}(k+1).
\end{align*}
The last two steps used (\ref{eq:temp2}) and (\ref{eq:temp4}). This verifies (\ref{eq:temp12}) for $t=k+1$. By the principle of induction, (\ref{eq:temp11})-(\ref{eq:temp12}) are verified for all $t=0,1,2,\ldots$. 

\vspace{-1mm}Finally, we take any $t=0,1,\ldots$ and proceed to prove (\ref{eq:eqiv}). Using (\ref{eq:temp11})-(\ref{eq:temp12}) and (\ref{eq7}), we obtain
\begin{align*}
P_{i\rightarrow j}(t) &= \Gamma_{ji} - \Psi_{ji}(\Gamma_{ij}+P_0(t))^{-1}\Psi_{ij} \\
&= \Gamma_{ji} - \Psi_{ji}\Sigma_{i\rightarrow j}(t-1)\Psi_{ij},
\end{align*}
which leads to 
\begin{align*}
P_i(t+1) &= P_{ii} + \sum_{v\in \mathcal{N}_i} P_{v\rightarrow i}(t) \\
&= P_{ii} +\sum_{v\in \mathcal{N}_i} \Gamma_{iv} - \sum_{v\in \mathcal{N}_i} \Psi_{iv}\Sigma_{v\rightarrow i}(t-1)\Psi_{vi} \\
&= \Psi_{ii} - \sum_{v\in \mathcal{N}_i} \Psi_{vi}^T\Sigma_{v\rightarrow i}(t-1)\Psi_{vi} 
= \hat{\Psi}_i(t),   
\end{align*}
which is the first part of (\ref{eq:eqiv}). Similarly, using (\ref{eq:temp11})-(\ref{eq:temp12}) and (\ref{eq8}), we obtain
\begin{align*}
&P_{i\rightarrow j}(t+1) \mu_{i\rightarrow j}(t+1) \\
=& - \Psi_{ji} (\Gamma_{ij}+P_0(t+1))^{-1} P_0(t+1)\mu_0(t+1) \\
=& - \Psi_{ji}\Sigma_{i\rightarrow j}(t-1)\Sigma_{i\rightarrow j}^{-1}(t)x_{i\rightarrow j}(t) 
= - \Psi_{ji}x_{i\rightarrow j}(t),
\end{align*}
which leads to the  the second part of (\ref{eq:eqiv}) as follows:
\begin{align*} 
\mu_i(t+1) &= P_i^{-1}(t+1)[P_{ii}\mu_{ii} +\sum_{v\in \mathcal{N}_i} P_{v\rightarrow i}(t)\mu_{v\rightarrow i}(t)] \\
&=P_i^{-1}(t+1) [\alpha_i - \sum_{v\in \mathcal{N}_i} \Psi_{iv}x_{v\rightarrow i}(t-1)]\\
&=\hat{\Psi}_i^{-1}(t) \hat{\alpha}_i(t) =\hat{x}_i(t).
\end{align*}
\end{proof}

\vspace{-4mm}\begin{rem}
{\rm Despite their equivalence, Algorithms 1 and 2 have some significant differences: 1) Algorithm 1 uses the self and edge measurements directly, whereas Algorithm 2 starts with the computed $\Psi$.  2) Algorithm 1 applies to vector variables and its more general version in \cite{damian2015} can work for measurements involving more than two variables, whereas the Gaussian BP algorithm is  for scalar variables with pairwise measurements only \cite{weiss2001}. }
\end{rem}

\vspace{-2mm} Next, we establish a crucial technical  property about the comparison matrix of $\Psi$ for the case of scalar variables.

\begin{lemma}\label{la1} Suppose Assumptions \ref{a1}-\ref{a2} hold and  all $x_i\in \mathbb{R}$. Then,  the comparison matrix $\bar{\Psi}$ (as defined in Definition~\ref{def1}) is positive definite.
\end{lemma}

\begin{proof} 
For any nonzero vector $x\in \mathbb{R}^n$, we have
\bea
&& x^T\bar{\Psi} x\nn\\
&=&\sum\limits_{i}(\Psi_{ii}x_i^2+\sum_{j\in \mathcal{N}_i}\bar{\Psi}_{ij}x_ix_j)\nn\\
&=&\sum\limits_{i}\Big[(A^2_iR_i^{-1}\hspace{-1mm}+\hspace{-2mm}\sum_{j\in \mathcal{N}_i}(B^2_{ij}R_{(i,j)}^{-1}x_i^2-|B_{ij}B_{ji}|R_{(i,j)}^{-1}x_ix_j)\Big]\nn\\
&=&\sum_{i}A^2_iR_i^{-1}x_i^2+\frac{1}{2}\sum_{i}\sum_{j\in \mathcal{N}_i}R_{(i,j)}^{-1}(|B_{ij}|x_i-|B_{ji}|x_j)^2.\nn
\eea
It is obvious from the above that $\bar{\Psi}$ is positive semi-definite.
Suppose there exists a nonzero vector $x$ satisfies $x^T\bar{\Psi} x=0$, then we have $A^2_iR_i^{-1}x_i^2=0$ and $R_{(i,j)}^{-1}(|B_{ij}|x_i-|B_{ji}|x_j)^2=0$ for $\forall i,j\in \mathcal{V}$. Based on Assumption \ref{a2}, there is at least one node $i$ with $A_i\neq0$, so $x_i=0$. For $j\in \mathcal{N}_i$, $x_j=0$ holds (because $B_{ji}\ne0$).
Similarly, $x_k=0$ holds for all $k\in \mathcal{N}_j$. The rest components of $x$ can be done in the same manner. Because the measurement graph is connected, as in Assumption \ref{a1}, we have $x_i=0$ for all $i\in \mathcal{V}$. This contradicts the assumption that $x\ne0$. So $\bar{\Psi}$ is positive definite.
\end{proof}

We now present our main result below.

\vspace{-2mm}\begin{theorem} \label{th2}
Based on Assumptions \ref{a1} and \ref{a2}, we have the following properties for Algorithm~\ref{al1}.
\vspace{-2mm}\begin{itemize}
\item If $\mathcal{G}$ is acyclic with diameter $d$, the estimate $\hat{x}_i(k)$ obtained by running Algorithm \ref{al1} converges to the exact value $x_i^{\star}$ in $d$ iterations for all $i\in \mathcal{V}$.
\item If $\mathcal{G}$ is cyclic and $x_i$ are all scalars, the estimate $\hat{x}_i(k)$ obtained by running Algorithm \ref{al1} is asymptotically convergent to the exact value $x_i^{\star}$ as $k\rightarrow\infty$, for all $i\in \mathcal{V}$.
\item  Moreover, the convergence rate of $\hat{x}_i(k)$ for a cyclic graph $\mathcal{G}$ with scalar variables is bounded as follows:
\begin{align}
|\hat{x}_i(k)-x_i^{\star}|&\leq\rho(\bar{\Omega})^kC \label{eq:bbb}
\end{align}
for each $i\in \mathcal{V}$, where $\Omega=(\omega_{ij})=I-D^{-1}\Psi$, $\bar{\Omega}=(|\omega_{ij}|)$, $D=\mathrm{diag}\{\Psi_{ii}\}$, and $C>0$ is a constant.
\end{itemize}
\end{theorem}

\begin{proof}
Using Lemma 1, we know that $\Psi=H^TR^{-1}H$ is positive definite. In particular, the matrix $H$ has full column rank, and recall that $R$ is invertible. It follows that Assumptions 2 and 11 in \cite{damian2015} hold, which in turn means that Theorem~11 of \cite{damian2015} holds. More specifically, $\hat{x}_i(k)=x_i^{\star}$ for all $k\ge d_i$, where $d_i$ is the radius of node $i$ which is defined as the maximum distance between node $i$ and any other node in the graph. It follows that $\hat{x}_i(k)=x_i^{\star}$ for all $k\ge d$ because $d\ge d_i$ for all $i$. 
\\
Next, we show that $\Psi$ is generalized diagonally dominant under Assumptions \ref{a1}-\ref{a2} and scalar variables. Indeed, from Lemma 1, we have $\bar{\Psi}_{ij}\leq 0$ for $\forall i,j\in\{1,\cdots,n\}$ with $i\ne j$, and $\bar{\Psi}$ is positive definite. Then, $\bar{\Psi}$ is M-matrix according to \cite {berman1994}. Thus, $\bar{\Psi}$ is generalized diagonally dominant (see (M35), Th6.2.3 in \cite {berman1994}). It means that there exists a diagonal matrix $D=\mathrm{diag}\{d_1,\cdots,d_n\}$, $\forall d_i>0$, such that $\bar{\Psi}D$ is strictly diagonally dominant. That is, for each $i\in\{i,\cdots,n\}$,
$d_i\Psi_{ii}>\sum_{j\neq i}|\Psi_{ij}|d_j.$
Then, $\Psi D$ is strictly diagonally dominant, i.e. $\Psi$ is generalized diagonally dominant.
\\
Finally, since $\Psi$ is generalized diagonally dominant, the asymptotic convergence result for a general (loopy) graph with scalar variables follows from \cite{M}, and the convergence rate result follows from the work \cite{rate}. Indeed, the convergence rate for this algorithm in \cite{rate} is presented for $Ax = b$ when $A$ is symmetric. We can obtain the result by substituting $\Psi$ for $A$.
\end{proof}

\vspace{-5mm}\begin{rem}
{\rm Note that for any node $i$, the information from a far-away node $j$ is gradually passed on to node $i$ through neighbourhood communication. If we view the iterations as a dynamic process, the estimate $\hat{x}_i(t)$ in (\ref{xxx}) is an estimate of the global optimal solution $x_i^{\star}$ conditioned on the filtration generated by  the measurements from all the nodes that are within $t$ hops away from node $i$. 
The asymptotic convergence result in Theorem~\ref{th2} shows that, for edge measurements, this optimality holds asymptotically, i.e., $\hat{x}_i(t)$ is indeed the optimal estimate of $x_i^{\star}$ conditioned on this filtration as $t \rightarrow \infty$. }
\end{rem}

\vspace{-1mm}\section{Example}
Consider a loopy network with $13$ nodes in Fig. \ref{fig1}. There are two nodes without self measurement, i.e., for $i\in\{2,5\}$, $A_i=0$. The other nonzero $A_i$ and all the $B_{ij}$ are chosen randomly. Fig. \ref{fig2} shows the estimation error by Algorithm~\ref{al1}, where $y_1$ denotes the error measure defined by 
$y_1 = \log_{10}\{\sum_{i}(\hat{x}_i(k)-x_i^{\star})^2/n\}.$
  Also shown is the convergence rate bound (\ref{eq:bbb}). 
  We see that the convergence rate of Algorithm~\ref{al1} is faster than the rate of $\rho(\bar{\Omega})$.  

 \begin{figure}[htbp]
  \begin{center}
  \includegraphics[width=0.15\textwidth]{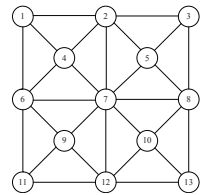}
   \caption{The network with 13 nodes} \label{fig1}
  \end{center}
\end{figure}

 \begin{figure}[htbp]
  \begin{center}
  \includegraphics[width=0.4\textwidth]{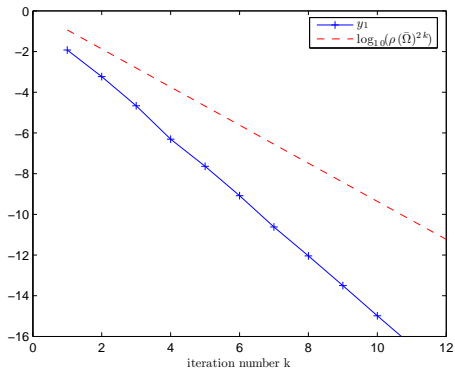}
   \caption{Estimation error  by Algorithm \ref{al1}} \label{fig2}
  \end{center}
\end{figure}

\section{Conclusion}
We have studied a fast distributed algorithm for the WLS estimation problem for a linear measurement network. We have provided an interpretation of this algorithm using the Gaussian BP algorithm, when only self measurements and edge measurements are involved. For scalar variables, we show that this algorithm computes asymptotically the correct (globally optimal) WLS solution for a general  network graph . We conjecture that similar properties hold for vector variables, but its analysis is challenging because couplings within a vector variables also need to be considered and that no results can be borrowed from  Gaussian BP in this case.


\begin{thebibliography}{}

\bibitem{kar2012} Kar, S., {\em et. al.} (2012). Distributed parameter estimation in sensor networks: nonlinear observation models and imperfect communication. {\it IEEE Trans.~Infor. Theory}, 58(6), 3575-3605.
\bibitem{lijl2007} Li, J.,  \& AlRegib, G. (2007). Rate-constrained distributed estimation in wireless sensor networks. {\it IEEE Trans.~Signal Proc.}, 55(5), 1634-1643.
\bibitem{mou2015} Mou, S. , Liu, J., \& Morse, A. S. (2015). A distributed algorithm for solving a linear algebraic equation. {\it IEEE Trans.~Auto. Control}, 60(11), 2863-2878.
\bibitem{damian2015} Marelli, D. E., \& Fu, M. (2015). Distributed weighted least-squares estimation with fast convergence for large-scale systems. {\it Automatica}, 51, 27-39.
\bibitem{lin2014} Lin, Z., Wang, L., Han, Z., \& Fu, M. (2014). Distributed formation control of multi-agent systems using complex Laplacian. {\it IEEE Trans.~Auto.~Control}, 59(7), 1765-1777.
\bibitem{lin2016} Lin, Z., Wang, L., Han, Z., \& Fu, M. (2016). A graph Laplacian approach to coordinate-free formation stabilization for directed networks. {\it EEE Trans.~Auto.~Control}, 61(5), 1269-1280.
\bibitem{nedic2009} Nedic, A., \& Ozdaglar, A. (2009). Distributed sub-gradient methods for multi-agent optimization. {\it IEEE Trans.~Auto.~Control}, 54(1), 48-61.
\bibitem{weiss2001} Weiss, Y., \& Freeman, W. T. (2001). Correctness of belief propagation in Gaussian graphical models of arbitrary topology. {\it Neural Computation}, 13(10), 2173-2200.
\bibitem{pearl1988} Pearl, J. (1988). {\it Probabilistic Reasoning in Intelligent Systems}. Morgan Kaufman.
\bibitem{M} Malioutov, D. M., Johnson, J. K., \& Willsky, A. S. (2006). Walk-sums and belief propagation in Gaussian graphical models. {\it Journal of Machine Learning Research}, 7, 2031-2064.
\bibitem{rate} Zhang, Z, \& Fu, M. (2019). On convergence rate of the Gaussian belief propagation algorithm for Markov networks. Submitted to {\it IEEE Trans.~Control of Network Systems}. arXiv:1903.02658.
\bibitem{berman1994} Berman, A., \& Plemmons, R. J. (1994). {\it Nonnegative Matrices in the Mathematical Sciences}. Classics Appl. Math., SIAM.

\end{thebibliography}
\end{document}